\documentclass[prodmode,permissions]{acmsmall-ec16}
\usepackage[numbers,sort&compress]{natbib} 
\usepackage{tcolorbox}
\usepackage{amsmath}
\usepackage{amsfonts}
\usepackage{verbatim}

\usepackage{cite}

\newtheorem{observation}{Observation}
\usepackage{todonotes}
\usepackage{cancel}

\usepackage{xspace}

\newcounter{alg}
\newcommand{\alg}[1]{\refstepcounter{alg}\label{#1}}

\newcommand{\PPAD}{\ensuremath{\mathtt{PPAD}}\xspace}

\newcommand{\reals}{\mathbb{R}}

\newcommand{\eps}{\ensuremath{\epsilon}\xspace}

\newcommand{\profx}{\ensuremath{\mathbf{x}}\xspace}
\newcommand{\prof}{\ensuremath{\mathbf{x}}\xspace}
\newcommand{\profy}{\ensuremath{\mathbf{y}}\xspace}
\newcommand{\profu}{\ensuremath{\mathbf{u}}\xspace}

\newcommand{\bp}{\ensuremath{\mathbf{p}}\xspace}
\newcommand{\bq}{\ensuremath{\mathbf{q}}\xspace}
\newcommand{\bbb}{\ensuremath{\mathfrak{b}}\xspace}
\newcommand{\ff}{\ensuremath{\mathfrak{f}}\xspace}
\newcommand{\lgame}{\ensuremath{\mathfrak{L}}\xspace}

\newcommand{\biased}{\ensuremath{\mathcal{B}}\xspace}

\newcommand{\regret}{\ensuremath{\mathcal{R}}\xspace}
\newcommand{\calh}{\ensuremath{\mathcal{H}}\xspace}

\newcommand{\calm}{\ensuremath{\mathcal{M}}\xspace}
\newcommand{\call}{\ensuremath{\mathcal{L}}\xspace}
\newcommand{\calp}{\ensuremath{\mathcal{P}}\xspace}
\newcommand{\calt}{\ensuremath{\mathcal{T}}\xspace}
\newcommand{\plip}{\ensuremath{\mathcal{P_{\lambda}}}\xspace}

\newcommand{\bmax}{\ensuremath{p_{\max}}\xspace}

\newcommand{\jj}{\ensuremath{\mathfrak{j}}\xspace}

\doi{XXXXXXX.XXXXXXX}

\begin{document}

\title{Lipschitz Continuity and Approximate Equilibria
}
\author{
ARGYRIOS DELIGKAS
\affil{University of Liverpool, UK}
JOHN FEARNLEY
\affil{University of Liverpool, UK}
PAUL SPIRAKIS
\affil{University of Liverpool, UK and Computer Technology Institute (CTI),
Greece}
}

\begin{abstract}
In this paper, we study games with continuous action spaces and non-linear
payoff functions. Our key insight is that Lipschitz continuity of the payoff
function allows us to provide algorithms for finding approximate equilibria in
these games. We begin by studying Lipschitz games, which encompass, for example,
all concave games with Lipschitz continuous payoff functions. We provide an
efficient algorithm for computing approximate equilibria in these games. Then we
turn our attention to penalty games, which encompass biased games and games in
which players take risk into account. Here we show that if the penalty function
is Lipschitz continuous, then we can provide a quasi-polynomial time
approximation scheme. Finally, we study distance biased games, where we present
simple strongly polynomial time algorithms for finding best responses in $L_1$,
$L_2^2$, and $L_\infty$ biased games, and then use these algorithms to provide
strongly polynomial algorithms that find $2/3$, $5/7$, and $2/3$ approximations
for these norms, respectively.
\end{abstract}




\maketitle

\section{Introduction} 

The Nash equilibrium~\cite{N} is the central solution concept that is studied in
game theory. However, recent advances have shown that computing an \emph{exact}
Nash equilibrium is \PPAD-complete~\cite{CDT,DGP}, and so there are unlikely to
be polynomial time algorithms for this problem. The hardness of computing exact
equilibria has lead to the study of \emph{approximate} equilibria: while an
exact equilibrium requires that all players have no incentive to deviate from
their current strategy, an $\epsilon$-approximate equilibrium requires only that
their incentive to deviate is less than $\epsilon$. 

A fruitful line of work has developed studying the best approximations that can
be found in polynomial-time for \emph{bimatrix games}, which are two-player
strategic form games. There, after a number of
papers~\cite{DMP07,DMP,BBM10}, the best known algorithm was
given by~\citet{TS}, who provide a polynomial time algorithm that finds a
$0.3393$-equilibrium. A prominent open problem is whether there exists a PTAS
for this problem. The existence of an FPTAS was ruled out by~\citet{CDT} unless
$\PPAD = \mathtt{P}$. While the existence of a PTAS remains open, there is
however a \emph{quasi-polynomial} approximation scheme given by~\citet{LMM}.

In a strategic form game, the game is specified by giving each player a finite
number of strategies, and then specifying a table of payoffs that contains one
entry for every possible combination of strategies that the players might pick.
The players are allowed to use mixed strategies, and so ultimately the payoff
function is a convex combination of the payoffs given in the table. However,
some games can only be modelled in a more general setting where the action
spaces are continuous, or the payoff functions are non-linear.

For example, Rosen's seminal work~\cite{Rosen65} considered a more general
setting of games, called \emph{concave games}, where each player picks a vector
from a convex set. The payoff to each player is specified by a function that
satisfies the following condition: if every other player's strategy is fixed,
then the payoff to a player is a convex function over his strategy space. Rosen
proved that concave games always have an equilibrium. A natural subclass of
concave games, studied by~\citet{CKP}, is the
class of biased games. A biased game is defined by a strategic form game, a
\emph{base strategy} and a \emph{penalty function}. The players play the
strategic form game as normal, but they all suffer a penalty for deviating from
their base strategy. This penalty can be a non-linear function, such as the
$L_2^2$ norm.

In this paper, we study the computation of approximate equilibria in such games.
Our main observation is that 
Lipschitz continuity of the players' payoff functions allows us to provide
algorithms that find approximate equilibria. 
Several papers have studied how the Lipschitz continuity of the players' payoff 
functions affects the existence, the quality, and the complexity of the 
equilibria of the underlying game. \citet{AS13} studied many 
player games and derived bounds for the Lipschitz constant of the utility 
functions for the players that guarantees the existence of pure approximate 
equilibrium for the game. 
\citet{DP14} proved that anonymous games posses pure
approximate equilibria whose quality depends on the Lipschitz constant of the
payoff functions and the number of pure strategies the players have and proved
that this approximate equilibrium can be computed in polynomial time.
Furthermore, they gave a polynomial-time approximation scheme for anonymous
games with many players and constant number of pure strategies. 
\citet{Bab13} presented a best-reply dynamic for $n$ players Lipschitz anonymous
games with two strategies which reaches an approximate pure equilibrium in 
$O(n \log n)$ steps. Recently, \citet{CDO15} proved that it is \PPAD-complete
to compute an \eps-equilibrium in anonymous games with seven pure strategies,
when \eps is exponentially small in the number of the players.
\citet{DK15} studied how some variants of the Lipschitz
continuity of the utility functions are sufficient to guarantee hindsight 
stability of equilibria.

\subsection{Our contribution.}

\paragraph{\bf Lipschitz games}
We begin by studying a very general class of games, where each player's strategy
space is continuous, and represented by a convex set of vectors, and where the
only restriction is that the payoff function is Lipschitz continuous. This class
encompasses, for example, every concave game in which the payoffs are Lipschitz
continuous. This class is so general that exact equilibria, and even approximate
equilibria may not exist. Nevertheless, we give an efficient algorithm that
either outputs an $\epsilon$-equilibrium, or determines that game has no exact
equilibria. More precisely, for $M$ player games that are $\lambda$-continuous
in the $L_p$ norm, for $p \ge 2$, and where $\gamma = \max \|\profx\|_p$ over
all $\profx$ in the strategy space, we either compute an $\epsilon$-equilibrium
or determine that no exact equilibrium exists in time $O\left(Mn^{Mk+l}\right)$,
where $k = O\big(\frac{\lambda^2Mp\gamma^2}{\eps^2}\big)$ and $l =
O\big(\frac{\lambda^2p\gamma^2}{\eps^2}\big)$. Observe that this is a polynomial
time algorithm when $\lambda$, $p$, $\gamma$, $M$, and $\eps$ are constant.

To prove this result, we utilize a recent result of \citet{B15}, which
states that for every vector in a convex set, there is another vector that is
$\epsilon$ close to the original in the $L_p$ norm, and is a convex combination
of $b$ points on the convex hull, where $b$ depends on $p$ and $\epsilon$, but
does not depend on the dimension. Using this result, and the Lipschitz
continuity of the payoffs, allows us to reduce the task of finding an
$\epsilon$-equilibrium to checking only a small number of strategy profiles, and
thus we get a brute-force algorithm that is reminiscent of the QPTAS given by
\citet{LMM} for bimatrix games.

However, life is not so simple for us. Since we study a very general class of
games, verifying whether a given strategy profile is an $\epsilon$-equilibrium
is a non-trivial task. It requires us to compute a \emph{regret} for each
player, which is the difference between the player's best response payoff and
their actual payoff. Computing a best response in a bimatrix game is trivial,
but for Lipschitz games, computing a best response may be a hard problem. We get
around this problem by instead giving an algorithm to compute \emph{approximate}
best responses. Hence we find \emph{approximate} regrets, and it turns out that
this is sufficient for our algorithm to work.

\paragraph{\bf Penalty games}

We then turn our attention to \emph{penalty games}. In these games, the players
play a strategic form game, and their utility is the payoff achieved in the game
\emph{minus} a penalty. The penalty function can be an arbitrary function that
depends on the player's strategy. This is a general class of games that
encompasses a number of games that have been studied before. The biased games
studied by~\citet{CKP}, are penalty games where the penalty is determined by the
amount that a player deviates from a specified base strategy.
The biased model was studied in the past by psychologists~\cite{TK74} and it is
close to what they call \emph{anchoring} \cite{Kahneman92,Chapman99}. 
%
In their seminal paper, \citet{FP10} introduced a model for \emph{risk prone} 
games. This model resembles penalty games since the risk component can be 
encoded in the penalty function. \citet{MM15} followed this line of research 
and provided results on the complexity of deciding if such games possess an 
equilibrium.

We again show that Lipschitz continuity helps us to find approximate equilibria.
The only assumption that we make is that the penalty function is Lipschitz
continuous in an $L_p$ norm with $p \ge 2$. Again, this is a weak restriction,
and it does not guarantee that exact equilibria exist. Even so, we give a
quasi-polynomial time algorithm that either finds an $\epsilon$-equilibrium, or
verifies that the game has no exact equilibrium.

Our result can be seen as a generalisation of the QPTAS given by \citet{LMM} for
bimatrix games. Their approach is to show the existence of an approximate
equilibrium with a logarithmic support. They proved this via the probabilistic
method: if we know an exact equilibrium of a bimatrix game, then we can take
logarithmically many samples from the strategies, and with positive probability
playing the sampled strategies uniformly will be an approximate equilibrium.

We take a similar approach, but since our games are more complicated, our proof
is necessarily more involved. In particular, for \citet{LMM}, proving that the
sampled strategies are an approximate equilibrium only requires showing that the
expected payoff is close the payoff of a pure best response. In penalty games,
best response strategies are not necessarily pure, and so the events that we
must consider are more complex.

\paragraph{\bf Distance biased games}

Finally, we consider distance biased games, which are a subclass of penalty
games that have been studied recently by~\citet{CKP}. They showed that, under
very mild assumptions on the bias function, biased games always have an exact
equilibrium. Furthermore, for the case where the bias function is either the
$L_1$ norm, or the $L_2^2$ norm, they give an exponential time algorithm for
finding an exact equilibrium.

Our results for penalty games already give a QPTAS for biased games, but we are
also interested in whether there are polynomial-time algorithms that can find
non-trivial approximations. We give a positive answer to this question for games
where the bias is the $L_1$ norm, the $L_2^2$ norm, or the $L_\infty$ norm. We
follow the well-known approach of~\citet{DMP}, who gave a simple algorithm for
finding a $0.5$-approximate equilibrium in a bimatrix game. Their approach is as
follows: start with an arbitrary strategy $\profx$ for player 1, compute a best
response $j$ for player 2 against $\profx$, and then compute a best response $i$
for player 1 against $j$. Player 1 mixes uniformly between $\profx$ and $i$,
while player 2 plays $j$.

We show that this algorithm also works for biased games, although the
generalisation is not entirely trivial. Again, this is because best responses
cannot be trivially computed in biased games. For the $L_1$ and $L_\infty$
norms, best responses can be computed via linear programming, and for the
$L_2^2$ norm, best responses can be formulated as a quadratic program, and it
turns out that this particular QP can be solved in polynomial time by the
ellipsoid method. However, none of these algorithms are strongly polynomial. We
show that, for each of the norms, best responses can be found by a simple
strongly-polynomial combinatorial algorithm. We then analyse the quality of
approximation provided by the technique of~\citet{DMP}. We obtain a strongly
polynomial algorithm for finding a $2/3$ approximation in $L_1$ and $L_\infty$
biased games, and a strongly polynomial algorithm for finding a $5/7$
approximation in $L_2^2$ biased games. For the latter result, in the special
case where the bias function is the inner product of the player's strategy we
find a $13/21$ approximation.

\section{Preliminaries}
\label{sec:pre}

We start by fixing some notation. For each positive integer $n$ we use $[n]$ to 
denote the set $\{1, 2, \ldots, n\}$, we use $\Delta^n$ to denote the 
$(n-1)$-dimensional simplex, and $\|x\|_p$ to denote the $p$-norm of a vector 
$x \in \reals^d$, i.e. $\|x\|_p = \left(\sum_{i \in [d]}|x_i|^p\right)^{1/p}$. 
Given a set $X = \{x_1, x_2, \ldots, x_n\} \subset \reals^d$, we use $conv(X)$
to denote the convex hull of $X$.

\paragraph{\bf Games and strategies}
A game with $M$-players can be described by a set of available actions for each
player and a utility function for each player that depends both on his chosen 
action and the actions the rest of the players chose. For each player $i \in [M]$
we use $S_i$ to denote his set of available actions and we call it \emph{strategy
space}. We will use $x_i \in S_i$ to denote a specific action chosen by player
$i$ and we will call it as the \emph{strategy} of player $i$. Furthermore, we 
use $\prof = (x_1, \ldots, x_M)$ to denote a \emph{strategy profile} of the game.
We use $T_i(x_i, \prof_{-i})$ to denote the utility of player $i$ when he plays
the strategy $x_i$ and the rest of the players play according to the strategy 
profile $\prof_{-i}$. A strategy $\hat{x}_i$ is a \emph{best response} against 
the strategy profile $\prof_{-i}$, if $T_i(\hat{x}_i, \prof_{-i}) \geq 
T_i(x_i, \prof_{-i})$ for all $x_i \in S_i$. The \emph{regret} player $i$ suffers
under a strategy profile $\prof$ is the difference between the utility of his 
best response and his utility under \prof, i.e. $T_i(\hat{x}_i, \prof_{-i}) - 
T_i(x_i, \prof_{-i})$.

\paragraph{\bf $\lambda_p$-Lipschitz Games}
We will use the notion of the $\lambda_p$-Lipschitz continuity.
\begin{definition}[$\lambda_p$-Lipschitz]
\label{def:lip}
A function $f: A \rightarrow \reals$, with $A \subseteq \reals^d$ is 
$\lambda_p$-Lipschitz continuous if for every $x$ and $y$ in $A$, it is true that 
$|f(x) - f(y)| \leq \lambda \cdot \|x-y\|_p$.
\end{definition}
We call the game $\lgame := (M, n, \lambda, p, \gamma, \calt)$ 
\emph{$\lambda_p$-Lipschitz} if for each player $i \in [M]$ 
\begin{itemize}
\item the strategy space $S_i$ is the convex hull of $n$ vectors $y_1, \ldots, 
y_n$ in $\reals^d$,
\item $\max_{x_i \in S_i}\|x_i\|_p \leq \gamma$
\item the utility function $T_i(\prof) \in \calt$ is $\lambda_p$-Lipschitz continuous.
\end{itemize}

\paragraph{\bf Two Player Penalty Games}
A two player penalty game \calp is defined by a tuple 
$\big(R, C, \ff_r(\profx), \ff_c(\profy) \big)$, where $(R, C)$ is a bimatrix 
game and $\ff_r(\profx)$ and $\ff_c(\profy)$ are the penalty 
functions for the row and the column player respectively. 
The utilities for the players under a strategy profile $(\profx, \profy)$, 
denoted by $T_r(\profx, \profy)$ and $T_c(\profx, \profy)$, are given by
\begin{align*}
T_r(\profx, \profy) = \profx^TR\profy - \ff_r(\profx) \qquad \qquad
T_c(\profx, \profy) = \profx^TC\profy - \ff_c(\profy).
\end{align*}
We will use \plip to denote two player penalty games with $\lambda_p$-Lipschitz 
penalty functions. 
%
A special class of penalty games is when $\ff_r(\profx) = \profx^T\profx$ and 
$\ff_c(\profy) = \profy^T\profy$. We call these games as \emph{inner product} 
penalty games.
%

\paragraph{\bf Two Player Biased Games}
This is a subclass of penalty games, where extra constraints are added to the 
penalty functions $\ff_r(\profx)$ and $\ff_c(\profy)$ of the players. In this 
class of games there is a \emph{base strategy} and for each player and the 
penalty they receive is increasing with the distance between the strategy they 
choose and their base strategy. Formally, the row player has a base strategy
$\bp \in \Delta^n$, the column player has a base strategy $\bq$ and their 
strictly increasing penalty functions are defined as $\ff_r(\|\profx - \bp\|^s_t)$
and $\ff_c(\|\profy - \bq\|^l_m)$ respectively.

\paragraph{\bf Two Player Distance Biased Games}
This is a special class of biased games where the penalty function is a fraction
of the distance between the base strategy of the player and his chosen strategy.
Formally, a two player distance biased game \biased is defined by a tuple 
$\big( R,C, \bbb_r(\profx, \bp), \bbb_c(\profy, \bq), d_r, d_c \big)$, where 
$(R,C)$ is a bimatrix game, $\bp \in \Delta^n$ is a base strategy for the row 
player, $\bq \in \Delta^n$ is a base strategy for the column player, 
$\bbb_r(\profx, \bp) = \| \profx - \bp\|^s_t$ and 
$\bbb_c(\profy, \bq) = \| \profy - \bq\|^l_m$ are penalty functions for the row
and the column player respectively.  
The utilities for the players under a strategy profile $(\profx, \profy)$, 
denoted by $T_r(\profx, \profy)$ and $T_c(\profx, \profy)$, are given by
\begin{align*}
T_r(\profx, \profy) = \profx^TR\profy - d_r \cdot \bbb_r(\profx, \bp) \qquad \qquad
T_c(\profx, \profy) = \profx^TC\profy - d_c \cdot \bbb_c(\profy, \bq)
\end{align*}
where $d_r$ and $d_c$ are non negative constants.


%
%

\paragraph{\bf Solution Concepts} 
The standard solution concept in game theory is the notion of equilibrium. A 
strategy profile is an equilibrium if no player can increase his utility by 
unilaterally changing his strategy. A relaxed version of this concept 
is the approximate equilibrium, or \eps-equilibrium.
Intuitively, a strategy profile is an \eps-equilibrium if no player can increase 
his utility more than \eps by unilaterally changing his strategy. Formally, a 
strategy profile \prof is an \eps-equilibrium in a game \lgame if for every 
player $i \in [M]$ it holds that
\begin{align*}
T_i(x_i, \prof_{-i}) \geq T_i(x'_i, \prof_{-i})- \eps \quad \text{for all $x_i' \in S_i$}.
\end{align*}

In~\cite{CDT} it was proven that, unless $\mathtt{P} = \PPAD$, there is no FPTAS 
for computing an $\eps$-NE in bimatrix games. The same result holds for the
class of penalty games where the penalty functions $\ff$ for the players depend 
on $n$, the size of the underlying bimatrix game, and 
$\lim_{n\rightarrow \infty}\ff = 0$ for every player. Let $\calp'$ to denote 
this class of games.
\begin{theorem}
\label{thm:nofptas}
Unless $\mathtt{P} = \PPAD$, there is no FPTAS for computing an $\eps$-equilibrium 
in penalty games in $\calp'$.
\end{theorem}
\begin{proof}
For the sake of contradiction suppose that there is an FPTAS for computing an
\eps-equilibrium for penalty games in $\calp'$. Then given an $n \times n$ 
bimatrix game $(R,C)$, define the penalty game $\big( R,C, 
\ff_r(\profx), \ff_c(\profy) \big)$ from the family $\calp'$ where 
$\lim_{n\rightarrow \infty}\ff_r(\profx) = 0$ and 
$\lim_{n\rightarrow \infty}\ff_c(\profy) = 0$. Let $(\profx^*, \profy^*)$ be an 
\eps-equilibrium for the penalty game. This means that for all 
$\profx' \in \Delta^n$ it holds that $\profx^{*^T}R\profy^* - \ff_r(\profx^*) \geq 
\profx'^TR\profy^* - \ff_r(\profx') - \eps$ or, equivalently, 
$\profx^{*^T}R\profy^* \geq \profx'^TR\profy^* - \eps'$, where 
$\eps' = \eps + \ff_r(\profx^*) - \ff_r(\profx')$. 
Similarly, $\profx^{*^T}C\profy^* \geq \profx^{*^T}C\profy' - \eps''$, where 
$\eps'' = \eps + \ff_c(\profy^*) - \ff_r(\profy')$. But $\eps' = \eps'' = \eps$
when $n \rightarrow \infty$. Hence $(\profx^*, \profy^*)$ is a $\eps$-NE for 
the bimatrix game $(R,C)$. This means that if there is an FPTAS for computing an
$\eps$-equilibrium in a penalty game in $\calp'$ then there is an FPTAS for 
computing an $\eps$-NE in $(R,C)$ which is a contradiction, unless 
$\mathtt{P} = \PPAD$.
\qed
\end{proof}

\section{Approximate equilibria in $\lambda_p$-Lipschitz games}

\label{sec:lip-games}
In this section, we give an algorithm for computing
approximate equilibria in $\lambda_p$ Lipschitz games. Note that, our definition
of a $\lambda_p$-Lipschitz game does not guarantee that an equilibrium always
exists. Our technique can be applied \emph{irrespective} of whether an exact
equilibrium exists. If an exact equilibrium does exist, then our technique will
always find an $\eps$-equilibrium. If an exact
equilibrium does not exist, then our then our algorithm either finds an
$\eps$-equilibrium or reports that the game does not have an exact equilibrium.

We will utilize the following theorem that was recently proved in~\citet{B15}.
\begin{theorem}[\cite{B15}]
\label{thm:barman}
Given a set of vectors $X = \{x_1, x_2, \ldots, x_n \} \subset \reals^d$, let 
$conv(X)$ denote the convex hull of $X$. Furthermore, let
$\gamma := \max_{x \in X}\| x\|_p$ for some $2 \leq p < \infty$.
For every $\eps > 0$ and every $\mu \in conv(X)$, there exists an 
$\frac{4p\gamma^2}{\eps^2}$ uniform vector $\mu' \in conv(X)$ such that 
$\|\mu - \mu' \|_p \leq \eps$.
\end{theorem}

If we combine the Theorem~\ref{thm:barman} with the Definition~\ref{def:lip} we
get the following lemma.

\begin{lemma}
\label{lem:optim}
Let $X = \{x_1, x_2, \ldots, x_n \} \subset \reals^d$, let 
$f: conv(X) \rightarrow \reals$ be a $\lambda_p$-Lipschitz continuous function 
for some $2 \leq p < \infty$, let $\eps>0$ and let 
$k = \frac{4\lambda^2p\gamma^2}{\eps^2}$, where 
$\gamma := \max_{x \in X}\| x\|_p$. Furthermore, let $f(\prof^*)$ be the optimum 
value of $f$. Then we can compute a $k$-uniform point $\prof' \in conv(X)$ in 
time $O(n^k)$, such that $|f(\prof^*) - f(\prof')|< \eps$. 
\end{lemma}
\begin{proof}
From Theorem~\ref{thm:barman} we know that for the chosen value of $k$ there 
exists a $k$-uniform point $\prof'$  such that 
$\|\prof' - \prof^*\|_p < \eps/\lambda$. 
Since the function $f(\profx)$ is $\lambda_p$-Lipschitz continuous, we get that 
$|f(\prof') - f(\prof^*)| < \eps$. 
In order to compute this point we have to exhaustively evaluate the function 
$f$ in all $k$-uniform points and choose the point that it maximizes/minimizes 
its value. Since there are ${n+k-1 \choose k} = O(n^k)$ possible $k$-uniform 
points, the theorem follows.
\qed
\end{proof}

We now prove our result about Lipschitz games. In what follows we will study a
$\lambda_p$-Lipschitz game $\lgame := (M, n, \lambda, p, \gamma, \calt)$.
Assuming the existence of an exact Nash equilibrium, we establish the existence
of a $k$-uniform approximate equilibrium in the game \lgame, where $k$ depends
on $M,\lambda, p$ and $\gamma$. Note that $\lambda$ depends heavily on $p$ and
the utility functions for the players.

Since by the definition of $\lambda_p$-Lipschitz games the 
strategy space $S_i$ for every player $i$ is the convex hull of $n$ vectors
$y_1, \ldots, y_n$ in $\reals^d$, any $x_i \in S_i$ can be written as a convex 
combination of $y_j$s. Hence, $x_i = \sum_{j=1}^n \alpha_jy_j$, where 
$\alpha_j > 0$ for every $j \in [n]$ and $\sum_{j=1}^n \alpha_j = 1$. Then, 
$\alpha = (\alpha_1, \ldots, \alpha_n)$ is a probability distribution over the 
vectors $y_1, \ldots, y_n$, i.e. vector $y_j$ is drawn with probability 
$\alpha_j$. Thus, we can sample a strategy $x_i$ by the probability distribution
$\alpha$.

So, let $\prof^*$ be an equilibrium for \lgame and let $\prof'$ be a sampled
uniform strategy profile from $\prof^*$. For each player $i$ we define the
following events
\begin{align}
\label{ev:payoff}
\phi_i & = \big\{|T_i(x'_i, \prof'_{-i}) - T_i(x^*_i, \prof^*_{-i}) | < \eps/2 \big\}\\
\label{ev:equil}
\pi_i & = \big\{T_i(x_i, \prof'_{-i}) < T_i(x'_i, \prof'_{-i})  + \eps \big\} 
\quad \text{for all possible $x_i$}\\
\label{ev:close}
\psi_i & = \left\{\|x'_i - x^*_i \|_p < \frac{\eps}{2M\lambda} \right\} \quad 
\text{for some $p > 0$}.
\end{align}
Notice that if all the events $\pi_i$ occur at the same time, then the sampled 
profile $\prof'$ is an $\eps$-equilibrium. 
We will show that if for a player $i$ the events $\phi_i$ and $\bigcap_j\psi_j$ 
hold, then the event $\pi_i$ has to be true too.
\begin{lemma}
\label{lem:con-pi}
For all $i \in [M]$ it holds that 
$\bigcap_{j \in [M]} \psi_j \cap \phi_i \subseteq  \pi_i$.
\end{lemma}
\begin{proof}
Suppose that both events $\phi_i$ and $\bigcap_j \psi_{j \in [M]}$ hold. We will 
show that the event $\pi_i$ must be true too. Let $x_i$ be an arbitrary strategy, 
let $\prof^*_{-i}$ be a strategy profile for the rest of the players, and let 
$\prof'_{-i}$ be a sampled strategy profile from $\prof^*_{-i}$. 
Since we assume that the events $\psi_j$ is true for all $j$ we get
$\|\prof'_{-i} - \prof^*_{-i}\|_p \leq \sum_{j\neq i} \|x'_j - x^*_j\|_p $
we get that
\begin{align*}
\|\prof'_{-i} - \prof^*_{-i}\|_p & \leq \sum_{j\neq i} \|x'_j - x^*_j\|_p \\
 & \leq \sum_{j\neq i} \frac{\eps}{2M\lambda} \\
 & \leq \frac{\eps}{2\lambda}.
\end{align*}
Furthermore, since by assumption the utility functions for the players are
$\lambda_p$-Lipschitz continuous we have that
$$\big|T_i(x_i, \prof'_{-i}) - T_i(x_i,\prof^*_{-i}) \big| \leq \frac{\eps}{2}.$$
This means that
\begin{align}
\nonumber
T_i(x_i, \prof'_{-i}) & \leq T_i(x_i,\prof^*_{-i}) + \frac{\eps}{2} \\
\label{eq:qph1}
& \leq T_i(x^*_i,\prof^*_{-i}) + \frac{\eps}{2}
\end{align}
since $T_i(x^*_i,\prof^*_{-i}) \geq T_i(x_i, \prof^*_{-i})$ for all possible $x_i$;
the strategy profile $(x^*_i,\prof^*_{-i})$ is an equilibrium of the game.
Furthermore, since by assumption the event $\phi_i$ is true we get that
\begin{align}
\label{eq:qph2}
T_i(x^*_i, \prof^*_{-i}) < T_i(x'_i, \prof'_{-i}) + \frac{\eps}{2}.
\end{align}
Hence, if we combine the inequalities~\eqref{eq:qph1} and~\eqref{eq:qph2} we get 
that $T_i(x_i, \prof'_{-i}) < T_i(x'_i, \prof'_{-i}) + \eps$ for all possible 
$x_i$. Thus, if the events $\phi_i$ and $\psi_j$ for every $j \in [M]$ hold, 
then the event $\pi_i$ holds too.
\qed
\end{proof}

We are ready to prove the main result of the section.
\begin{theorem}
\label{thm:existence}
In any game $\lambda_p$-Lipschitz game \lgame that posses an equilibrium and any 
$\eps > 0$, there is a $k$-uniform strategy profile, with 
$k = \frac{16M^2\lambda^2p\gamma^2}{\eps^2}$ that is an \eps-equilibrium.
\end{theorem}
\begin{proof}
In order to prove the claim, it suffices to show that there is a strategy 
profile where every player plays a $k$-uniform strategy, for the chosen value of
$k$, such that the events $\pi_i$ hold for all $i \in [M]$. Since the utility 
functions in \lgame are $\lambda_p$-Lipschitz continuous it holds that 
$\bigcap_{i \in [n]} \psi_i \subseteq \bigcap_{i \in [n]}\phi_i$. Furthermore, 
combining that with the Lemma~\ref{lem:con-pi} we get that 
$\bigcap_{i \in [n]} \psi_i \subseteq \bigcap_{i \in [n]}\pi_i$. 
Thus, if the event $\psi_i$ is true for every $i \in [n]$, then the event 
$\bigcap_{i \in [n]} \pi_i$ is true as well. 

From the Theorem~\ref{thm:barman} we get that for each $i \in [M]$ there is a 
$\frac{16M^2\lambda^2p\gamma^2}{\eps^2}$-uniform point $x'_i$ such that the 
event $\psi_i$ occurs with positive probability. The claim follows.
\qed
\end{proof}

Theorem~\ref{thm:existence} establishes the existence of a $k$-uniform
approximate equilibrium, but this does not immediately give us our approximation
algorithm. The obvious approach is to perform a brute force check of all
$k$-uniform strategies, and then output the one the provides the best
approximation. There is a problem with this, however, since computing the
quality of approximation requires us to compute the regret for each player,
which in turn requires us to compute a best response for each player. Computing
an exact best response in a Lipschitz game is a hard problem in general, since
we make no assumptions about the utility functions of the players. Fortunately,
it is sufficient to instead compute an \emph{approximate} best response for each
player, and  Lemma~\ref{lem:optim} can be used to do this. The following Lemma
is a consequence of Lemma~\ref{lem:optim}.

%
\begin{lemma}
\label{lem:eps-br}
Let \prof be a strategy profile for a $\lambda_p$-Lipschitz game \lgame, and let
$\hat{x}_i$ be a best response for the player $i$ against the profile 
$\prof_{-i}$. There is a $\frac{4\lambda^2p\gamma^2}{\eps^2}$-uniform strategy 
$x_i'$ that is an \eps-best response against $\prof_{-i}$, i.e. 
$|T_i(\hat{x}_i, \prof_{-i}) - T_i(x'_i, \prof_{-i})| < \eps$.
\end{lemma}
%
%

Our goal is to \emph{approximate} the approximation guarantee for a given
strategy profile. More formally, given a strategy profile $\profx$ that is an
$\epsilon$-equilibrium, and a constant $\delta > 0$, we want an algorithm that
outputs a number within the range $[\epsilon - \delta, \epsilon + \delta]$.
Lemma~\ref{lem:eps-br} allows us to do this. For a given strategy profile
$\profx$, we first compute $\delta$-approximate best responses for each player,
then we can use these to compute $\delta$-approximate regrets for each player.
The maximum over the $\delta$-approximate regrets then gives us an approximation
$\epsilon$ with a tolerance of $\delta$. This is formalised in the following
algorithm.



%
\alg{alg:approx}
\begin{tcolorbox}[title=Algorithm~\ref{alg:approx}. Evaluation of approximation guarantee]
\textbf{Input:} A strategy profile \prof for \lgame, and a constant $\delta > 0$.\\
\textbf{Output:} An additive $\delta$-approximation of the approximation guarantee $\alpha(\prof)$ for the strategy profile \prof.
\begin{enumerate}
\itemsep 1.2mm
\item Set $l = \frac{4\lambda^2p\gamma^2}{\delta^2}$.
\item For every player $i \in [M]$
\begin{enumerate}
\label{step:iter}
\item For every $l$-uniform strategy $x'_i$ of player $i$ compute 
$T_i(x'_i, \prof_{-i})$.
\label{step:maxpay}
\item Set $m^* = \max_{x'_i}T_i(x'_i, \prof_{-i})$.
\item Set $\regret_i(\prof) = m^* - T_i(x_i,\prof_{-i})$.
\label{step:regi}
\end{enumerate}
\item Set $\alpha(\prof) = \delta + \max_{i \in [M]}\regret_i(\prof)$.
\label{step:apxg}
\item Return $\alpha(\prof)$.
\end{enumerate}
\end{tcolorbox}

Utilising the above algorithm, we can now produce an algorithm to find an
approximate equilibrium in Lipschitz games. The algorithm checks all $k$-uniform
strategy profiles, using the value of $k$ given by 
Theorem~\ref{thm:existence}, and for each one, computes an approximation of the
quality approximation using the algorithm given above. 

\alg{alg:lipschitz}
\begin{tcolorbox}[title=Algorithm~\ref{alg:lipschitz}. 3\eps-equilibrium for $\lambda_p$-Lipschitz game \lgame]
\textbf{Input:} Game \lgame and $\eps>0$.\\
\textbf{Output:} An 3\eps-equilibrium for \lgame.
\begin{enumerate}
\itemsep 1.2mm
\item Set $k > \frac{16\lambda^2Mp\gamma^2}{\eps^2}$.
\label{step:qptasiter}
\item For every $k$-uniform strategy profile $\prof'$
\label{step:apxeval}
\begin{enumerate}
\item Compute an $\epsilon$-approximation of $\alpha(\prof')$.
\item If the $\epsilon$-approximation of $\alpha(\prof')$ is less than  $2\eps$, return $\prof'$.
\end{enumerate}
\end{enumerate}
\end{tcolorbox}
If the algorithm returns a strategy profile $\prof$, then it must be a $3\eps$
equilibrium. This is because we check that an $\epsilon$-approximation of
$\alpha(\prof)$ is less than $2 \eps$, and therefore $\alpha(\prof) \leq 3
\eps$. Secondly, we argue that if the game has an exact Nash equilibrium, then
this procedure will always output a $3\eps$-approximate equilibrium. From
Theorem~\ref{thm:existence} we know that if $k >
\frac{16\lambda^2Mp\gamma^2}{\eps^2}$, then there is a $k$-uniform strategy
profile $\profx$ that is an $\eps$-equilibrium for \lgame. When we apply our
approximate regret algorithm to $\profx$, to find an $\epsilon$-approximation of
$\alpha(\profx)$, the algorithm will return a number that is less than $2 \eps$,
hence $\profx$ will be returned by the algorithm.

To analyse the running time, observe that there are ${n+k-1 \choose k} =
O(n^k)$ possible $k$-uniform strategies for each player, thus $O(n^{Mk})$
$k$-uniform strategy profiles. Furthermore, our regret approximation algorithm
runs in time $O(Mn^l)$, where $l = \frac{4\lambda^2p\gamma^2}{\eps^2}$. Hence,
we get the next theorem.
\begin{theorem}
\label{thm:qptas-concave}
Given a $\lambda_p$-Lipschitz game \lgame that posses an equilibrium and any 
$\eps>0$, a 3\eps-equilibrium can be computed in time $O\left(Mn^{Mk+l}\right)$, 
where $k = O\big(\frac{\lambda^2Mp\gamma^2}{\eps^2}\big)$ and 
$l = O\big(\frac{\lambda^2p\gamma^2}{\eps^2}\big)$.
\end{theorem}

Notice that in might be computationally hard to decide whether a game posses an
equilibrium or not. Nevertheless, our algorithm can be applied in \emph{any} 
$\lambda_p$-Lipschitz game, without being affected by the existence or not of an
exact equilibrium. If the game does not posses an exact equilibrium then our 
algorithm either finds an approximate equilibrium or decides that there is no 
$k$-uniform strategy profile that is an \eps-equilibrium for the game, thus the 
game does not posses an exact equilibrium.

\begin{theorem}
\label{thm:qptas-general}
For any game $\lambda_p$-Lipschitz game \lgame in time $O\left(Mn^{Mk+l}\right)$, 
we can either compute a $3\eps$-equilibrium, or decide that \lgame does not posses an 
exact equilibrium, where $k = O\big(\frac{\lambda^2Mp\gamma^2}{\eps^2}\big)$ and 
$l = O\big(\frac{\lambda^2p\gamma^2}{\eps^2}\big)$.
\end{theorem}


\section{A quasi-polynomial algorithm for penalty games}
\label{sec:qptas}

In this section we present an algorithm that, for any $\eps>0$, can compute an
\eps-equilibrium for any penalty game in $\plip$ in quasi-polynomial time. For
the algorithm, we take the same approach as we did in the previous section for
Lipschitz games: We show that if an exact equilibrium exists, then a $k$-uniform
approximate equilibrium always exists too, and provide a brute-force search 
algorithm for finding it. Once again,
since best response computation may be hard for this class of games, we must
provide an approximation algorithm for finding the quality of an approximate
equilibrium. The majority of this section is dedicated to proving an appropriate
bound for $k$, to ensure that $k$-uniform approximate equilibria always exist.



We first focus on penalty games that posses an exact equilibrium.
So, let $(\profx^*, \profy^*)$ be an equilibrium of the game and let 
$(\profx', \profy')$ be a $k$-uniform strategy profile sampled from this 
equilibrium. We define the following four events:
\begin{align*}
\phi_r  = & \big\{|T_r(\profx', \profy') - T_r(\profx^*, \profy^*)| < \eps/2 \big\} \\
\pi_r = & \big\{ T_r(\profx, \profy') < T_r(\profx', \profy') + \eps \big \} \qquad \text{for all \profx}\\
\phi_c  = & \big\{|T_c(\profx', \profy') - T_c(\profx^*, \profy^*)| < \eps/2 \big\} \\
\pi_c = & \big\{ T_c(\profx', \profy) < T_c(\profx', \profy') + \eps \big \} \qquad \text{for all \profy}.
\end{align*}
The goal is to derive a value for $k$ such that all the four events above are
true, or equivalently $Pr(\phi_r \cap \pi_r \cap \phi_c \cap \pi_r) > 0$.


Note that in order to prove that $(\profx', \profy')$ is an \eps-equilibrium we 
\emph{only} have to consider the events $\pi_r$ and $\pi_c$. Nevertheless, as 
we show in the Lemma~\ref{lem:intersection}, the events $\phi_r$ and $\phi_c$ 
are crucial in our analysis. The proof of the main theorem boils down to the the
events $\phi_r$ and $\phi_c$. Furthermore, proving that there is a $k$-uniform 
profile $(\profx', \profy')$ that fulfills the events $\phi_r$ and $\phi_c$ too, 
proves that the approximate equilibrium we compute approximates the utilities 
the players receive under an exact equilibrium too. 

In what follows we will focus only on the row player, since similar analysis can 
be applied for the column player too. Firstly we study the event $\pi_r$ and we 
show how we can relate it with the event $\phi_r$.  
\begin{lemma}
\label{lem:intersection}
For all penalty games 
it holds that $Pr(\pi_r^c) \leq n \cdot e^{-\frac{k\eps^2}{2}}+ Pr(\phi_r^c)$.
\end{lemma}
\begin{proof}
We begin by introducing the following auxiliary events for all $i \in [n]$
\begin{align*}
\psi_{ri} = \big\{R_i\profy' < R_i\profy^* + \frac{\eps}{2} \big\}.
\end{align*}
We prove how the events $\psi_{ri}$ and the event $\phi_r$ are related with 
the event $\pi_r$.
Assume that the event $\phi_r$ and the events $\psi_{ri}$ for all $i \in [n]$ 
are true . 
Let \profx be any mixed strategy for the row player. 
Since by assumption $R_i\profy' < R_i\profy^* + \frac{\eps}{2}$ and since \profx
is a probability distribution, it holds that
$\profx^TR\profy' < \profx^TR\profy^* + \frac{\eps}{2}$. 
If we subtract $\ff_r(\profx)$ from each side we get that 
$\profx^TR\profy' - \ff_r(\profx) < \profx^TR\profy^* - \ff_r(\profx) + \frac{\eps}{2}$.
This means that $T_r(\profx, \profy') < T_r(\profx, \profy^*) + \frac{\eps}{2}$
for all \profx. But we know that 
$T_r(\profx, \profy^*) \leq T_r(\profx^*, \profy^*)$ for all $\profx \in \Delta^n$,
since $(\profx^*, \profy^*)$ is an equilibrium. Thus, we get that 
$T_r(\profx, \profy') < T_r(\profx^*, \profy^*) + \frac{\eps}{2}$ for all possible
\profx. Furthermore, since the event $\phi_r$ is true too, we get that 
$T_r(\profx, \profy') < T_r(\profx', \profy') + \eps$. Thus, if the events 
$\phi_r$ and $\psi_{ri}$ for all $i \in [n]$ are true, then the event $\pi_r$ 
must be true as well. 
Formally, $\phi_r \bigcap_{i \in [n]} \psi_{ri} \subseteq \pi_r$. 
Thus, $Pr(\pi_r^c) \leq Pr(\phi_r^c) + \sum_i \psi_{ri}$.
Using the Hoeffding bound, we get that $Pr(\psi_{ri}^c) \leq e^{-\frac{k\eps^2}{2}}$ 
for all $i \in [n]$. 
Our claim follows.
\qed
\end{proof}
With Lemma~\ref{lem:intersection} in hand, we can see that in order to compute
a value for $k$ it is sufficient to study the event $\phi_r$. We introduce the
following auxiliary events that we will study seperately:
\begin{align*}
\phi_{ru} & = \big\{ |\profx'^TR\profy' - \profx^{*^T}R\profy^*| < \eps/4 \big\}\\
\phi_{r\bbb} & = \big\{ |\ff_r(\profx') - \ff_r(\profx^*)| < \eps/4 \big\}.
\end{align*}
It is easy to see that if both $\phi_{r\bbb}$ and $\phi_{ru}$are true, then the
event $\phi_r$ must be true too, formally 
$\phi_{r\bbb} \cap \phi_{ru} \subseteq \phi_r$. 
Using the analysis from~\cite{LMM} we can prove that 
$Pr(\phi_{ru}^c) \leq 2e^{-\frac{k\eps^2}{8}}$. Thus, it remains to study the 
the event $\phi^c_{r\bbb}$. 
\begin{lemma}
\label{lem:rbp-bound}
$\Pr(\phi^c_{r\bbb}) \leq \frac{8\lambda\sqrt{p}}{\eps\sqrt{k}}$.
\end{lemma}
\begin{proof}
Since we assume that the penalty function $\ff_r(\profx')$ is $\lambda_p$-Lipschitz
continuous the event $\phi_{r\bbb}$ can be replaced by the event 
$\phi_{r\bbb'} = \big\{ \|\profx' - \profx^*\|_p < \eps/4\lambda \big\}$. It is
easy to see that $\phi_{r\bbb} \subseteq \phi_{r\bbb'}$. Then, using the proof 
of Theorem~2 from~\cite{B15} we get that 
$E[\|\profx' - \profx^*\|_p] \leq \frac{2\sqrt{p}}{\sqrt{k}}$. Thus, using 
Markov's inequality we get that 
\begin{align*}
Pr(\| \profx' - \profx^* \|_p \geq \frac{\eps}{4\lambda}) & \leq 
\frac{E[\|\profx' - \profx^*\|_p]}{\frac{\eps}{4\lambda}} \\
& \leq \frac{8\lambda\sqrt{p}}{\eps \sqrt{k}}.
\end{align*}
\qed
\end{proof}

%
%

%
We are ready to prove our theorem
\begin{theorem}
\label{thm:qptas1}
For any equilibrium $(\profx^*, \profy^*)$ of a penalty game from the class 
$\plip$, any $\eps>0$, and any $k \in \frac{\Omega(\lambda^2 \log n)}{\eps^2}$, 
there exists a $k$-uniform strategy profile $(\profx', \profy')$ that:
\begin{enumerate}
\item $(\profx', \profy')$ is an \eps-equilibrium for the game,
\item $|T_r(\profx', \profy') - T_r(\profx^*, \profy^*)| < \eps/2$,
\item $|T_c(\profx', \profy') - T_c(\profx^*, \profy^*)| < \eps/2$.
\end{enumerate}
\end{theorem}
\begin{proof}
Let us define the  event $GOOD =  \phi_r \cap \phi_c \cap \pi_{r} \cap \pi_{c}$.
In order to prove our theorem it suffices to prove that $Pr(GOOD)>0$.
Notice that for the events $\phi_c$ and $\pi_c$ we can use the same analysis as 
for $\phi_r$ and $\pi_r$ and get the same bounds. 

Thus, using 
Lemma~\ref{lem:intersection} and the analysis for the events $\phi_{ru}$ and 
$\phi_{r\bbb}$ we get that
\begin{align*}
Pr(GOOD^c) & \leq Pr(\phi_r^c) + Pr(\pi_r^c) + Pr(\phi_c^c) + Pr(\pi_c^c)\\
 & \leq 2\big( Pr(\phi_r^c) + Pr(\pi_r^c) \big)\\
 & \leq 2\big( 2 Pr(\phi_r^c) + n\cdot e^{-\frac{k\eps^2}{2}} \big) \quad \text{(from Lemma~\ref{lem:intersection})}\\
 & \leq 2\big( 2 Pr(\phi_{ru}^c) + 2Pr(\phi_{r\bbb'}^c) + n\cdot e^{-\frac{k\eps^2}{2}} \big)\\
 & \leq 2\big( 4e^{-\frac{k\eps^2}{8}} + \frac{8\lambda\sqrt{p}}{\eps\sqrt{k}} + 
n\cdot e^{-\frac{k\eps^2}{2}} \big)  \quad \text{(from Lemma~\ref{lem:rbp-bound})}\\
 & < 1 \quad \text{for the chosen value of $k$}.
\end{align*}
Thus, $Pr(GOOD)>0$ and our claim follows.
\qed
\end{proof}

The Theorem~\ref{thm:qptas1} establishes the \emph{existence} of a $k$-uniform
strategy profile $(\profx', \profy')$ that is an $\eps$-equilibrium. However, as
with the previous section, we must provide an efficient method for approximating
the quality of approximation provided by a given strategy profile. To do so, we
first give the following lemma, which shows that approximate best responses can
be computed in quasi-polynomial time for penalty games.


\begin{lemma}
\label{lem:cruc}
Let $(\profx, \profy)$ be a strategy profile for a penalty game $\plip$, and let 
$\hat{\profx}$ be a best response against \profy. There is an $l$-uniform 
strategy $\profx'$, with $l = \frac{17\lambda^2\sqrt{p}}{\eps^2}$, that is an 
\eps-best response against \profy, i.e. 
$T_r(\hat{\profx}, \profy) < T_r(\profx', \profy) + \eps$.
\end{lemma}
\begin{proof}
We will prove that $|T_r(\hat{\profx}, \profy) - T_r(\profx', \profy)| < \eps$
which implies our claim. Let 
$\phi_1 = \{|\hat{\profx}^TR\profy - \profx'^TR\profy| \leq \eps/2\}$ and 
$\phi_2 = \{|\ff_r(\hat{\profx}) - \ff_r(\profx')| < \eps/2  \}$
Notice that Lemma~\ref{lem:rbp-bound} does not use anywhere the fact that 
$\profx^*$ is an equilibrium strategy, thus it holds even if $\prof^*$ is 
replaced by $\hat{\profx}$. Thus, 
$Pr(\phi_2^c) \leq \frac{4\lambda\sqrt{p}}{\eps\sqrt{k}}$. Furthermore, using
the analysis from~\cite{LMM} again, we can prove that $Pr(\phi_1^c) \leq 
2e^{-\frac{k\eps^2}{4}}$ and using similar arguments as in the proof of 
Theorem~\ref{thm:qptas1} it can be easily proved that for the chosen of $l$ it
holds that $Pr(\phi_1^c) + Pr(\phi_2^c) < 1$, thus the events $\phi_1$ and
 $\phi_2$ occur with positive probability and our claim follows.
\qed
\end{proof}

Having given this Lemma, we can reuse Algorithm~\ref{alg:approx}, but with $l$
set equal to $\frac{17\lambda^2\sqrt{p}}{\eps^2}$, to provide an algorithm that
aproximates the quality of approximation of a given strategy profile. Then, we 
can reuse Algorithm~\ref{alg:lipschitz} with 
$k = \frac{\Omega(\lambda^2 \log n)}{\eps^2}$ to provide a quasi-polynomial time
algorithm that finds approximate equilibia in penalty games. Notice again that
our algorithm can be applied in games that it is computationally hard to verify 
whether an exact equilibrium exists. Our algorithm either will compute an 
approximate equilibrium or it will fail to find one, thus it will decide that 
the game does not posses an exact equilibrium.

%
\begin{theorem}
\label{thm:qptas2}
In any penalty game $\plip$ with constant number of players and any $\eps > 0$, 
in quasi polynomial time we can either compute a $3\eps$-equilibrium, or decide 
that \plip does not posses an exact equilibrium.
\end{theorem}


\section{Distance Biased Games}

In this section, we focus on three particular classes of distance biased games,
and we provide polynomial-time approximation algorithms for these games.
We focus on the following three penalty functions:
\begin{itemize}
\itemsep2mm
\item \textbf{$L_1$ penalty}:  
$\bbb_r(\profx, \bp) = \|\profx - \bp\|_1 = \sum_i |\profx_i - \bp_i|$.
\item \textbf{$L_2^2$ penalty}: 
$\bbb_r(\profx, \bp) = \|\profx - \bp\|^2_2 = \sum_i (\profx_i - \bp_i)^2$.
\item \textbf{$L_\infty$ penalty}:
$\bbb_r(\profx, \bp) = \|\profx - \bp\|_\infty = \max_i |\profx_i - \bp_i|$.
\end{itemize}

Our approach is to follow the well-known technique of \citet{DMP} that finds a
$0.5$-NE in a bimatrix game. The algorithm that we will use for all three
penalty functions is given below.

\alg{alg:base}
\begin{tcolorbox}[title=Algorithm~\ref{alg:base}. The Base Algorithm]
\begin{enumerate}
\itemsep2mm
\item Compute a best response $\profy^*$ against \bp.
\item Compute a best response \profx against $\profy^*$.
\item Set $\profx^* = \delta\cdot\bp + (1-\delta)\cdot\profx$, 
for some $\delta \in [0,1]$.
\item Return the strategy profile $(\profx^*, \profy^*)$.
\end{enumerate}
\end{tcolorbox}

While this is a well-known technique for bimatrix games, note that it cannot 
immediately be applied to penalty games. This is because the algorithm requires
us to compute two best response strategies, and while computing a best-response
is trivial in bimatrix games, this is not the case for penalty games. Best
responses for $L_1$ and $L_\infty$ penalties can be computed in polynomial-time
via linear programming, and for $L_2^2$ penalties, the ellipsoid algorithm can
be applied. However, these methods do not provide strongly polynomial
algorithms.

In this section, for each of the penalties, we develop a simple combinatorial
algorithm for computing best response strategies for each of these penalties.
Our algorithms are strongly polynomial. Then, we  determine the quality of the
approximation given by the base algorithm when our best response techniques are
used. In what follows we make the common assumption that the payoffs of the 
underlying bimatrix game $(R, C)$ are in $[0, 1]$.


\subsection{A 2/3-approximation algorithm for $L_1$-biased games}

We start by considering $L_1$-biased games. Suppose that we want to compute a
best-response for the row player against a fixed strategy $\profy$ of the column
player. We will show that best response strategies in $L_1$-biased games have a 
very particular form: if $b$ is the best response strategy in the (unbiased) 
bimatrix game $(R, C)$, then the best-response places all of its probability on 
$b$ \emph{except} for a certain set of rows $S$ where it is too costly to shift
probability away from \bp. The rows $i \in S$ will be played with  $\bp_i$ to
avoid taking the penalty for deviating.

The characterisation for whether it is too expensive to shift away from $\bp$ is
given by the following lemma.


\begin{lemma}
\label{lem:l1br}
Let $j$ be a pure strategy, let $k$ be a pure strategy with $\bp_k > 0$, and let
$\profx$ be a strategy with $\profx_k = \bp_k$. The
utility for the row player increases when we shift probability from $k$ to $j$
if and only if $R_j\profy - R_k\profy - 2d_r > 0$. 
\end{lemma}
\begin{proof}
Suppose that we shift $\delta$ probability from $k$ to $j$, where $\delta \in
(0, \bp_k]$. Then the utility for the row player is equal to $T_r(\profx,\profy)
+ \delta\cdot(R_j\profy - R_k\profy - 2d_r)$, where the final term is the
penalty for shifting away from $k$. Thus, the utility for the row player
increases under this shift if and only if $R_j\profy - R_k\profy - 2d_r > 0$. 
\qed
\end{proof}
%
Observe that, if we are able to shift probability away from a strategy $k$, then
we should obviously shift it to a best response strategy for the (unbiased)
bimatrix game, since this strategy maximizes the increase in our payoff. Hence,
our characterisation of best response strategies is correct. This gives us the
following simple algorithm for computing best responses.
%
\alg{alg:l1}
\begin{tcolorbox}[title=Algorithm~\ref{alg:l1}. Best Response Algorithm for $L_1$ penalty]
\begin{enumerate}
\item Set $S = 0$.
\item Compute a best response $b$ against $\profy$ in the unbiased bimatrix game
$(R, C)$.
\item For each index $i \ne b$ in the range $1 \le i \le n$:
\begin{enumerate}
\item If $R_b\cdot\profy-R_i\cdot\profy-2d_r \leq 0$, then set $\profx_i =
\bp_i$ and $S = S + \bp_i$.
\item Otherwise set $\profx_i = 0$
\end{enumerate}
\item Set $\profx_b = 1 - S$.
\item Return $\profx$.
\end{enumerate}
\end{tcolorbox}

Our characterisation has a number of consequences.
Firstly, it can be seen that if $d_r \geq 1/2$, then there is no profitable 
shift of probability between any two pure strategies, since $0 \leq R_i\profy \leq 1$
for all $i \in [n]$. Thus, we get the following corollary.
\begin{corollary}
If $d_r \geq 1/2$, then \bp is a dominant strategy.
\end{corollary}
Moreover, since we can compute a best response in polynomial time we get the next
theorem.
\begin{theorem}
\label{thm:l1easy}
In biased games with $L_1$ penalty functions and $\max\{ d_r, d_c\} \geq 1/2$, 
an equilibrium can be computed in polynomial time.
\end{theorem}

Finally, using the characterization of best responses we can see that there is 
a connection between the equilibria of the distance biased game and the well 
supported Nash equilibria (WSNE) of the underlying bimatrix game.
\begin{theorem}
\label{obs:lone}
Let $\biased =\big( R,C, \bbb_r(\profx, \bp), \bbb_c(\profy, \bq), d_r, d_c \big)$  
be a distance biased game with $L_1$ penalties and let $d := \max \{d_r, d_c\}$.
Any equilirbium of \biased is a $2d$-WSNE for the bimatrix game $(R, C)$.
\end{theorem}
\begin{proof}
Let $(\profx^*, \profy^*)$ be an equilibrium for \biased. From the best response
Algorithm for $L_1$ penalty games we can see that $\profx^*_i > 0$ if and only 
if $R_b\cdot\profy^*-R_i\cdot\profy^*-2d_r \leq 0$, where $b$ is a pure best
response against $\profy^*$. This means that for every $i \in [n]$ with 
$\profx^*_i > 0$, it holds that 
$R_i\cdot\profy^* \geq \max_{j \in [n]}R_j\cdot\profy^*-2d$. Similarly, it 
holds that $C^T_i\cdot\profx^* \geq \max_{j \in [n]}C^T_j\cdot\profx^*-2d$ for 
all $i \in [n]$ with $\profy^*_i > 0$. This is the definition of a $2d$-WSNE for 
the bimatrix game $(R, C)$.
\qed
\end{proof}

\subsubsection{Approximation algorithm}

We now analyse the approximation guarantee provided by the base
algorithm for $L_1$-biased games. So, let $(\profx^*, \profy^*)$ be the strategy
profile the is returned by the base algorithm. Since we have already shown that
exact Nash equilibria can be found in games with either $d_c \ge 1/2$ or $d_r
\ge 1/2$, we will assume that both $d_c$ and $d_r$ are less than $1/2$, since
this is the only interesting case. 

We start by considering the regret of the row player. The following lemma will
be used in the analysis of all three of our approximation algorithms. 

\begin{lemma}
\label{lem:row_regret}
Under the strategy profile $(\profx^*, \profy^*)$ the regret for the row player
is at most $\delta$.
\end{lemma}
\begin{proof}
Notice that for all $i \in [n]$ we have
\begin{equation*}
|\delta\bp_i + (1-\delta)\profx_i  - \bp_i| = (1-\delta)|\profx_i  - \bp_i|, 
\end{equation*}
hence $\| \profx^* - \bp\|_1 = (1-\delta)\|\profx - \bp\|_1$ and 
$\| \profx^* - \bp\|_\infty = (1-\delta)\|\profx - \bp\|_\infty$. Furthermore,
notice that $\sum_i\big((1-\delta)\profx_i + \delta\bp_i -\bp_i\big)^2 = 
(1-\delta)^2\|\profx - \bp\|^2_2$, thus 
$\| \profx^* - \bp\|_2^2 \leq (1-\delta)\|\profx - \bp\|_2^2$.
Hence the payoff for the row player it holds
$T_r(\profx^*, \profy^*) \geq
\delta \cdot T_r(\bp, \profy^*) + (1-\delta) \cdot T_r(\profx, \profy^*)$
and his regret under the strategy profile $(\profx^*, \profy^*)$ is
\begin{align*}
\regret^r(\profx^*, \profy^*) & = \max_{\tilde{\profx}}T_r(\tilde{\profx}, \profy^*) - T_r(\profx^*, \profy^*)\\
& =T_r(\profx, \profy^*) - T_r(\profx^*, \profy^*) 
\qquad \text{(since \profx is a best response against $\profy^*$)}\\
& \leq \delta \big( T_r(\profx, \profy^*) - T_r(\bp, \profy^*) \big)\\
& \leq \delta \qquad \qquad \text{(since $\max_\profx T_r(\profx, \profy^*) \leq 1$ 
and  $T_r(\bp,\profy^*) \geq 0$)}.
\end{align*}
\qed
\end{proof}

Next, we consider the regret of the column player. The following lemma will be
used for both the $L_1$ case and the $L_\infty$ case. Observe that in the $L_1$
case, the precondition of $d_c\cdot \bbb_c(\profy^*,\bq) \leq 1$ always holds,
since we have $\|\profy^* - \bq\|_1 \leq 2$, thus $d_c\cdot \bbb_c(\profy^*,\bq)
\leq 1$ since we are only interested in the case where $d_c \leq 1/2$.

\begin{lemma}
\label{lem:col_l1inf}
If $d_c\cdot \bbb_c(\profy^*,\bq) \leq 1$, then under strategy profile
$(\profx^*, \profy^*)$ the column player suffers at most $2-2\delta$ regret.
\end{lemma}
\begin{proof}
The regret of the column player under the strategy profile 
$(\profx^*, \profy^*)$ is
\begin{align*}
\regret^c(\profx^*, \profy^*) & = \max_{\profy}T_c(\profx^*, \profy) - T_c(\profx^*, \profy^*)\\
& = \max_{\profy}\Big\{ (1-\delta) T_c(\profx, \profy) + \delta T_c(\bp, \profy) 
\big)\Big\} - (1-\delta) T_c(\profx, \profy^*) -\delta T_c(\bp, \profy^*)\\
& \leq (1-\delta) \big( \max_{\profy}T_c(\profx^*, \profy) - T_c(\profx, \profy^*)\big)
\text{(since $\profy^*$ is a best response against \bp)}\\
& \leq (1-\delta)(1 + d_c\cdot \bbb_c(\profy^*,\bq))\quad \text{(since $\max_\profx T_c(\profx^*, \profy) \leq 1$)}\\
& \leq (1-\delta)\cdot 2 \quad \text{(since $d_c\cdot \bbb_c(\profy^*,\bq) \leq 1$)}.
\end{align*}
\qed
\end{proof}

To complete the analysis, we must select a value for $\delta$ that equalises the
two regrets. It can easily be verified that setting $\delta = 2/3$ ensures that
$\delta = 2 - 2\delta$, and so we have the following theorem.
%
\begin{theorem}
\label{thm:l1apx}
In biased games with $L_1$ penalties a 2/3-equilibrium can be computed in
polynomial time.
\end{theorem}

\subsection{A 5/7-approximation algorithm for $L_2^2$-biased games}

We now turn our attention to biased games with an $L_2^2$ penalty. Again, we
start by giving a combinatorial algorithm for finding a best response.
Throughout this section, we fix $\profy$ as a column player strategy, and we
will show how to compute a best response for the row player.

Best responses in $L_2^2$-biased games can be found by solving a quadratic
program, and actually this particular quadratic program can be solved via the
ellipsoid algorithm~\cite{Kozlov80}. We
will give a simple combinatorial algorithm that uses the Karush-Kuhn-Tucker 
(KKT) conditions, and produces a closed formula for the solution. Hence, we will
obtain a strongly polynomial time algorithm for finding best responses.


Our algorithm can be applied on $L_2^2$ penalty functions and any value $d_r$,
but for notation simplicity we describe our method for $d_r = 1$. Furthermore,
we define $\alpha_i := R_i\profy + 2\bp_i$ and we call $\alpha_i$ as the payoff 
of pure strategy $i$. Then, the utility for the row player can be written as 
$T_r(\profx, \profy) = \sum_{i = 1}^n\profx_i\cdot\alpha_i - \sum_{i = 1}^n\profx^2_i - \bp^T\bp$.
Notice that the term $\bp^T\bp$ is a constant and it does not affect the solution
of the best response; so we can exclude it from our computations.
Thus, a best response for the row player against strategy \profy is the solution 
of the following quadratic program
\begin{align*}
\text{maximize} \qquad & \sum_{i = 1}^n\profx_i\cdot\alpha_i - \sum_{i = 1}^n\profx^2_i\\
\text{subject to} \quad & \sum_{i=1}^n\profx_i = 1\\
&  \profx_i \geq 0 \quad \text{ for all } i \in [n].
\end{align*}
The Lagrangian function for this problem is 
$$\mathcal{L}(\profx, \profy, \lambda, \profu) = \sum_{i = 1}^n\profx_i\cdot\alpha_i 
- \sum_{i = 1}^n\profx^2_i  - \lambda(\sum_{i = 1}^n\profx_i - 1) - 
\sum_{i=1}^n u_i\profx_i$$
and the corresponding KKT conditions
\begin{align}
\label{eq:c1}
\alpha_i - \lambda -2\profx_i -\profu_i = 0 & \quad \text{for all } i \in [n]\\
\label{eq:c2}
\sum_{i = 1}^n\profx_i = 1\\
\label{eq:c3}
\profx_i \geq 0 & \quad \text{for all } i \in [n]\\
\label{eq:comple}
\profx_i\cdot\profu_i = 0 & \quad \text{for all } i \in [n].
\end{align}
Constraints \eqref{eq:c1}-\eqref{eq:c3} are the stationarity conditions and 
\eqref{eq:comple} are the complementarity slackness conditions. We say that 
strategy \profx is a \emph{feasible response} if it satisfies the KKT conditions. 
The obvious way to compute a best response is by exhaustively checking all $2^n$ 
possible combinations for the complementarity conditions and choose the
feasible response that maximizes the utility for a player.
Next we prove how we can bypass the brute force technique and compute all best
responses in polynomial time. 

In what follows, without loss of generality, we assume that 
$\alpha_1 \geq \ldots \geq \alpha_n$. That is, the pure strategies are ordered
according to their payoffs. In the next lemma we prove that in every 
best response, if a player plays pure strategy $l$ with positive probability,
then he must play every pure strategy $k$ with $k < l$ with positive
probability.

\begin{lemma}
\label{lem:order}
In every best response $\profx^*$ if $\profx^*_l > 0$ then $\profx^*_k>0$ for all 
$k<l$.
\end{lemma}
\begin{proof}
For the sake of contradiction suppose that there is a best response $\profx^*$ 
and a $k<l$ such that $\profx^*_l>0$ and $\profx^*_k=0$. Let us denote 
$M = \sum_{i \neq \{l,k\}} \alpha_i \cdot \profx^*_i - \sum_{i \neq \{l,k\}} \profx^{*^2}_i $.
Suppose now that we shift some probability, denoted by $\delta$, from pure 
strategy $l$ to pure strategy $k$. Then his utility is 
$T_r(\profx^*, \profy) = M + \alpha_l \cdot(\profx^*_l-\delta) - (\profx^*_l-\delta)^2
+ \alpha_k \cdot \delta - \delta^2$, which is maximized for 
$\delta = \frac{\alpha_k-\alpha_l + 2\profx^*_l}{4}$. Notice that $\delta > 0$
since $\alpha_k \geq \alpha_l$ and $\profx^*_l > 0$, thus the row player can 
increase his utility by assigning positive probability to pure strategy $k$ which 
contradicts the fact that $\profx^*$ is a best response.
\qed
\end{proof}

Lemma~\ref{lem:order} implies that there are only $n$ possible supports that a
best response can use.
%
%
%
Indeed, we can exploit the KKT conditions to derive, for each candidate support,
the exact probability that each pure strategy would be played. We derive the
probability as a function of $\alpha_i$s and of the support size. Suppose that
the KKT conditions produce a feasible response when 
we set the support to have size $k$. From condition~\eqref{eq:c1} we get that 
$\profx_i = \frac{1}{2}(\alpha_i - \lambda)$ for all $1 \leq i \leq k$ and zero
else. But we know that $\sum_{j}^k \profx_j =1$. Thus we get that 
$\sum_{j = 1}^k \frac{1}{2}(\alpha_j - \lambda) = 1$ and if we solve for 
$\lambda$ get that $\lambda = \frac{\sum_{j = 1}^k \alpha_j - 2}{k}$. This means 
that for all $i \in [k]$ we get
\begin{align}
\label{eq:xi}
\profx_i = \frac{1}{2}\left(\alpha_i - \frac{\sum_{j = 1}^k \alpha_j - 2}{k}\right).
\end{align}

So, our algorithm does the following. It loops through all $n$ candidate
supports for a best response. For each one, it uses Equation~\eqref{eq:xi} to
determine the probabilities, and then checks whether these satisfy the KKT
conditions, and thus if this is a feasible response. If it is, then it is saved
for in a list of feasible responses, otherwise it is discarded. After all $n$
possibilities have been checked, the feasible response with the highest payoff
is then returned.
%
 
\alg{alg:l2}
\begin{tcolorbox}[title=Algorithm~\ref{alg:l2}. Best Response Algorithm for $L_2^2$ penalty]
\begin{enumerate}
\itemsep2mm
\item For $i = 1 \ldots n$
\begin{enumerate}
\item Set $\profx_1 \geq \ldots \geq \profx_i > 0$ and $\profx_{i+1} = \ldots = \profx_n = 0$.
\label{step2}
\item Check if there is a feasible response under these constraints.
\item If so, add it to the list of feasible responses.
\end{enumerate}
\item Among the feasible responses choose one with the highest utility.
\end{enumerate}
\end{tcolorbox}

\subsubsection{Approximation Algorithm}
We now show that the base algorithm gives a 5/7-approximation when applied to
$L_2^2$-penalty games. For the row player's regret, we can use
Lemma~\ref{lem:row_regret} to show that the regret is bounded by $\delta$.
However, for the column player's regret, things are more involved. We will show
that the regret of the column player is at most $2.5-2.5\delta$. That
analysis depends on the maximum entry of the base strategy $\bq$ and more
specifically on whether $\max_k \{\bq_k\} \leq 1/2$ or not.

\begin{lemma}
\label{lem:qleqhalf}
If $\max_k \{\bq_k\} \leq 1/2$, then the regret the column player suffers under
strategy profile $(\profx^*, \profy^*)$ is at most $2.5-2.5\delta$.
\end{lemma}
\begin{proof}
Note that when $\max_k \{\bq_k\} \leq 1/2$, then $\bbb_c=\|\profy-\bp\|_2^2 \leq
1.5$ for all possible \profy. Then, using the analysis from
Lemma~\ref{lem:col_l1inf}, along with the fact that $d_c\cdot
\bbb_c(\profy^*,\bq) \leq 2$ for $L_2^2$ penalties, and since by assumption $d_c
= 1$, the claim follows.
\qed
\end{proof}

For the case where there is a $k$ such that $\bq_k > 1/2$ a more involved 
analysis is needed. The first goal is to prove that under \emph{any} strategy 
$\profy^*$ that is a best response against $\bp$ the pure strategy $k$ is played
with positive probability. In order to prove that, first it is proven that there
is a feasible response against strategy \bp where pure strategy $k$ is played 
with positive probability. In what follows we denote $\alpha_i := C^T_i\bp + 2\bq_i$.

\begin{lemma} 
\label{lem:kktfeas}
Let $\bq_k>1/2$ for some $k \in [n]$. Then there is a feasible response where 
pure strategy $k$ is played with positive probability.
\end{lemma}
\begin{proof}
 Note that $\alpha_k > 1$ since by assumption 
$\bq_k > 1/2$. Recall from Equation~\eqref{eq:xi} that in a feasible response 
\profy it holds that $\profy_i = 
\frac{1}{2}\left(\alpha_i - \frac{\sum_{j = 1}^k \alpha_j - 2}{k}\right)$.

In order to prove the claim it is sufficient to show that $\profy_k>0$ when in 
the KKT conditions is set $\profy_i>0$ for all $i \in [k]$. Or equivalently, to 
show that $\alpha_k - \frac{\sum_{j = 1}^k \alpha_j - 2}{k}
= \frac{1}{k}\big((k-1)\alpha_k + 2 - \sum_{j = 1}^{k-1} \alpha_j\big)>0$. But,
\begin{align*}
(k-1)\alpha_k + 2 - \sum_{j = 1}^{k-1} \alpha_j & > k+1 - \sum_{j = 1}^{k-1} \big( C^T\profx + 2\bq_i \big) 
\quad \text{(since $\alpha_k > 1$)}\\
& \geq  k+1 - (k - 1) - \sum_{j = 1}^{k-1} 2\bq_i\\
& \geq 1 + \bq_k \quad \text{(since $\bq \in \Delta^n$)}\\
& > 0.
\end{align*}
The claim follows.
\qed
\end{proof}

Next it is proven that the utility of the column player is increasing when he 
adds pure strategies $i$ in his support such that $\alpha_i > 1$.
\begin{lemma}
\label{lem:monfeas}
Let $\profy^k$ and $\profy^{k+1}$ be two feasible responses with support size 
$k$ and $k+1$ respectively, where $\alpha_{k+1}>1$. 
Then $T_c(\profx, \profy^{k+1}) > T_c(\profx, \profy^k)$.
\end{lemma}
\begin{proof}
Let $\profy^k$ be a feasible response with support size $k$ for the column 
player against strategy \bp and let
$\lambda(k):=\frac{\sum_{j = 1}^k \alpha_j - 2}{2k}$. Then the utility of the 
column player when he plays $\profy^k$ can be written as
\begin{align*}
T_c(\profx, \profy^k) & = \sum_{i=1}^n \profy^k_i\cdot \alpha_i - \sum_{i=1}^n (\profx^k_i)^2 - \bq^T\bq\\
 & = \sum_{i=1}^k \profy^k_i \big(\alpha_i - \profy^k_i\big) - \bq^T\bq\\
 & = \sum_{i=1}^k \left(\frac{\alpha_i}{2} - \lambda(k)\right) \left(\frac{\alpha_i}{2} + \lambda(k)\right) - \bq^T\bq\\
 & = \frac{1}{4}\sum_{i =1}^k\alpha_i^2 - k\cdot\big(\lambda(k)\big)^2 - \bq^T\bq.
\end{align*}
The goal now is to prove that $T_c(\profx, \profy^{k+1}) - T_c(\profx, \profy^k) > 0$.
By the previous analysis for $T_c(\profx, \profy^k)$ and if 
$A := \sum_{i =1}^k\alpha_i - 2$, then
\begin{align*}
T_c(\profx, \profy^{k+1}) - T_c(\profx, \profy^k)
 & = \frac{1}{4}\sum_{i =1}^{k+1}\alpha_i^2 - (k+1)\big(\lambda(k+1)\big)^2
-\frac{1}{4}\sum_{i =1}^k\alpha_i^2 + k\cdot\big(\lambda(k)\big)^2\\
& = \frac{1}{4}\left( \alpha_{k+1}^2 + \frac{A^2}{k} - \frac{(A+\alpha_{k+1})^2}{k+1}\right)\\
& = \frac{1}{4}\left( \alpha_{k+1}^2 + \frac{1}{k+1}(A^2 - \alpha_{k+1}^2 - 2A\alpha_{k+1} )\right)\\
& = \frac{1}{4(k+1)} \big(k\alpha_{k+1}^2 + A^2 - 2A\alpha_{k+1} \big)\\
& > \frac{1}{4(k+1)} \big(k + A^2 - 2A \big) \quad \text{(since $1 <
\alpha_{k+1} \leq 2$ and $A > k-2$)} \\
& > \frac{1}{4(k+1)} \big(k^2 - 5k + 8 \big) \quad \text{(since $A > k-2$)} \\
& > 0.
\end{align*}
\qed
\end{proof}
Notice that $\alpha_k \geq 2\bp_k > 1$. Thus, the utility of the feasible 
response that assigns positive probability to pure strategy $k$ is strictly 
greater than the utility of any feasible responses that does not assign 
probability to $k$. Thus strategy $k$ is always played in a best response. 
Hence, the next lemma follows.
\begin{lemma}
\label{lem:pos-prob}
If there is a $k \in [n]$ such that $\bq_k > 1/2$, then in \emph{every} best 
response $\profy^*$ the pure strategy $k$ is played with positive probability.
\end{lemma}
Using now Lemma~\ref{lem:pos-prob} we can provide a better bound for the regret 
the column player suffers, since in every best response $\profy^*$ the 
pure strategy $k$ is played with positive probability.
\begin{lemma}
\label{lem:colbound2}
Let $\profy^*$ be a best response when there is a pure strategy $k$ with 
$\bq_k > 1/2$. Then the regret for the column player under strategy profile 
$(\profx^* ,\profy^*)$ is bounded by $2 - 2\delta$.
\end{lemma}
\begin{proof}
Before we proceed with our analysis we assume without loss of generality that 
$k=1$. Recall from the analysis for the Algorithm 1 that the regret for the 
column player is  
\begin{align}
\nonumber
\regret^c(\profx^*, \profy^*) & \leq (1-\delta) \Big(\max_{\tilde{\profy} \in \Delta} \{\hat{\profx}^TC\tilde{\profy} \} + 2\tilde{\profy}^T\bq_k - 2\profy^{*^T}\bq + \profy^{*^T}\profy^* \Big)\\
\label{eq:help1}
& \leq (1-\delta) \big(1+ 2\bq_k - 2\profy^{*^T}\bq + \profy^{*^T}\profy^*\big).
\end{align}
We focus now on the term $\profy^{*^T}\profy^* - 2\profy^{*^T}\bq$. 
It can be proven~\footnote{Appendix~\ref{app:proof}} that 
$\profy^{*^T}\profy^* - 2\profy^{*^T}\bq \leq 1 - 2\bq_k$. 
Thus, from~\eqref{eq:help1} we get that 
$\regret^c(\profx^*, \profy^*) \leq 2-2\delta$.
\qed
\end{proof}

Recall now that the regret for the row player is bounded by $\delta$, so if we 
optimize with respect to $\delta$ the regrets are equal for $\delta = 2/3$.
Thus, the next theorem follows, since when the there is a $k$ with $\bq_k > 1/2$ 
the Algorithm 1 produces a $2/3$-equilibrium. Hence, combining this with 
Lemma~\ref{lem:qleqhalf} the Theorem~\ref{thm:L2} follows for $\delta = 5/7$.

\begin{theorem}
\label{thm:L2}
In biased games with $L_2^2$ penalties a $5/7$-equilibrium can be computed 
in polynomial time.
\end{theorem}


\subsection{Inner product penalty games}

We observe that we can also tackle the case where the penalty function is the 
inner product of the strategy played, i.e. $\bp = \bq = \textbf{0}$. For these 
games, that we call inner product penalty games, we replace $\bp$ as the starting
point of the base algorithm with the fully mixed strategy $\profx^n$. Hence, for
that case $\profx^* = \delta\cdot\profx^n + (1-\delta)\cdot\profx$ for some
$\delta \in [0,1]$. In Appendix~\ref{sec:app-inner-appr} we prove the next 
theorem.
Again, the regret the row player suffers under strategy
profile $(\profx^*, \profy^*)$ is bounded by $\delta$.
\begin{lemma}
\label{lem:row-reg}
When the penalty function is the inner product of the strategy played, then the 
regret for the row player under strategy profile $(\profx^*, \profy^*)$ is 
bounded by $\delta$.
\end{lemma}
Furthermore, using similar analysis as in Lemma~\ref{lem:col_l1inf} it can be 
proven that the regret for the column player under strategy profile $(\profx^*, 
\profy^*)$ is bounded by $(1-\delta)(1+d_c\cdot\profy^{*^T}\profy^*)$. For the 
column player we will distinguish between the cases where $d_c \leq 1/2$ and 
$d_c > 1/2$. For the first case where $d_c \leq 1/2$ it is easy see that the 
algorithm produces a 0.6-equilibrium. For the other case, when $d_c > 1/2$,
first it is proven that there is no pure best response.
\begin{lemma}
\label{lem:pure}
If the penalty for the column player is equal to $\profy^T\profy$ and  
$d_c > \frac{1}{2}$, then there is no pure best response against any strategy 
of the row player.
\end{lemma}
\begin{proof}
Let $C_j$ to denote the payoff of the column player from his $j$-th pure strategy
against some strategy \profx played by the row player.
For the sake of contradiction, assume that there is a pure best response for the
column player where, without loss of generality, he plays only his first pure 
strategy. Suppose now that he shifts some probability to his second strategy, 
that is he plays the first pure  strategy with probability $x$ and the second 
pure strategy with probability $1-x$. The utility for the column player under 
this mixed strategy is $x\cdot C_1 + (1-x) \cdot C_2 -d_c\cdot( x^2 + (1-x)^2)$, 
which is maximized for $x = \frac{2d_c+C_1-C_2}{4d_c}$. Notice that $x > 0$, 
which means that the column player can deviate from the pure strategy and 
increase his utility. The claim follows.
\qed
\end{proof}

With Lemma~\ref{lem:pure} in hand, it can be proven that when $d_c > 1/2$ the 
column player does not play any pure strategy with probability greater than 3/4.
\begin{lemma}
\label{lem:ymax}
If $d_c>1/2$, then in $\profy^*$ no pure strategy is played with probability 
greater than 3/4.
\end{lemma}
\begin{proof}
For the sake of contradiction suppose that there is a pure strategy 
$i$ in $\profy^*$ that is played with probability greater than 3/4. Furthermore,
let $k$ be the support size of $\profy^*$. From Lemma~\ref{lem:pure}, since 
$d_c>1/2$, we know that there is no pure best response, thus $k \geq 2$. Then 
using Equation~\eqref{eq:xi} we get that 
$\frac{3}{4} < \frac{1}{2} \big( \alpha_i - \frac{\sum_{j=1}^k\alpha_j -2}{k}\big)$.
If we solve for $\alpha_j$ we get that $\alpha_i > \frac{3k-4}{2k-2} > 1$ which 
is a contradiction since when $\bq = \textbf{0}$ it holds that 
$\alpha_i = C^T_i\profx \leq 1$.
\qed
\end{proof}
A direct corollary from Lemma~\ref{lem:ymax} is that $\profy^{*^T}\profy^* \leq 5/8$.
Hence, we can prove the following lemma.

\begin{lemma}
\label{lem:col-reg}
Under strategy profile $(\profx^*, \profy^*)$ the regret for the column player is
bounded by $\frac{13}{8}(1-\delta)$.
\end{lemma}
\begin{proof}
Firstly, note that $T_c(\profx^*, \profy^*) = \delta\profx^{n^T}C\profy^* + 
(1-\delta)\profx^TC\profy^* - \profy^{*^T}\profy^*$. Moreover, 
$\max_{\tilde{\profy} \in \Delta} \{\profx^{n^T}C\tilde{\profy} - 
\tilde{\profy}^T\tilde{\profy}\}
- T_c(\profx^n, \profy^*) = 0$, since $\profy^*$ is a best response against 
$\profx^n$. Finally, notice that $0 \leq \profy^T\profy \leq 1$ for all \profy.
Thus, the regret for the column player is 
\begin{align*}
\regret^c(\profx^*, \profy^*) & = (1-\delta) \Big(\max_{\tilde{\profy} \in \Delta} \{\profx^TC\tilde{\profy} - 
\tilde{\profy}^T\tilde{\profy} \} - \profx^TC\profy^* + \profy^{*^T}\profy^* \Big)\\
& < (1-\delta) \big(1 + \frac{5}{8} \big).
\end{align*}
which matches the claimed result.
\qed
\end{proof}

If we combine Lemmas~\ref{lem:row-reg} and~\ref{lem:col-reg} and solve for 
$\delta$ we can see that the regrets are equal for $\delta = \frac{13}{21}$.
Thus, we get the following theorem for biased games where $\bq= \textbf{0}$.

\begin{theorem}
\label{thm:13}
The strategy profile $(\profx^*, \profy^*)$ is a $\frac{13}{21}$-equilibrium for 
biased games with $\bq = \textbf{0}$.
\end{theorem}

\subsection{A 2/3-approximation for $L_\infty$-biased games}

Finally, we turn our attention to the $L_\infty$ penalty. We start by giving a
combinatorial algorithm for finding best responses. Similar to the best response
Algorithm for the $L_1$ penalty, the intuition is to start from the base
strategy \bp of the row player and shift probability from pure strategies with
low payoff to pure strategies with higher payoff. This time though, the shifted
probability will be distributed between the pure strategies with higher payoff. 

Without loss of generality assume that $R_1\profy \geq \ldots \geq R_n\profy$,
ie., that the strategies are ordered according to their payoff in the unbiased
bimatrix game.
The set of pure strategies of the row player can be partitioned into three 
disjoint sets according to the payoff they yield:
\begin{align*}
\calh & := \{i \in [n]: R_i\profy = R_1\profy \} \\
\calm & := \{i \in ([n]\setminus \calh):  R_1\profy - R_i\profy - d_r < 0\}\\
\call & := \{i \in [n]: R_1\profy - R_i\profy -d_r > 0\}.
\end{align*}

Next we giver an algorithm that computes a best response for 
$L_\infty$ penalty.

\alg{alg:linf}
\begin{tcolorbox}[title=Algorithm~\ref{alg:linf}. Best Response Algorithm for $L_\infty$ penalty]
\begin{enumerate}
\itemsep2mm
\item For all $i \in \call$, set $x_i = 0$.
\item If $ \calp \leq |\calh|\cdot p_{\max}$, then set 
$x_i= \bp_i + \frac{\calp}{|\calh|}$ for all $i \in \calh$ and $x_j = \bp_j$ for
$j \in \calm$.
\item Else if $ \calp < |\calh \cup \calm| \cdot p_{\max}$, then
\begin{itemize}
\item Set $x_i= \bp_i + \bmax$ for all $i \in \calh$.
\item Set $k= \lfloor\frac{\calp - |\calh|\cdot\bmax}{\bmax}\rfloor$.
\item Set $x_i= \bp_i + \bmax$ for all $i \leq |\calh|+k$.
\item Set $x_{|\calh|+k+1}= \bp_{|\calh|+k+1} + \calp - (|\calh|+k)\cdot\bmax$.
\item Set $x_j= \bp_j$ for all $|\calh|+k+2 \leq j \leq |\calh|+|\calm|$.
\end{itemize}
\item Else set $x_i = \bp_i + \frac{\calp}{|\calh \cup \calm|}$ for all $i \in \calh \cup \calm$.
\end{enumerate}
\end{tcolorbox}

Let $\bmax := \max_{i \in \call} \bp_i$ and let $\calp := \sum_{i \in \call}\bp_i$. 
Then for every best response the following lemma holds.
\begin{lemma}
\label{lem:linfbr}
If $\call \neq \emptyset$, then for any best response \profx of the row player 
against strategy \profy it holds that $\| \profx - \bp \|_\infty \geq \bmax$.
Else \bp is the best response.
\end{lemma}
\begin{proof}
Using similar arguments as in Lemma~\ref{lem:l1br}, it can be proven that if 
there are no pure strategies $i$ and $k$ such that $R_k\profy - R_i\profy -d_r < 0$
then any shifting of probability decreases the utility of the row player. Thus,
the best response of the player is \bp. On the other hand, if there are strategies
$i$ and $k$ such that $R_k\profy - R_i\profy -d_r > 0$, then the utility of 
the row player increase if all the probability from strategy $i$ is shifted to 
pure strategy $k$. The set \call contains all these pure strategies. Let 
$\jj \in \call$ be the pure strategy that defines \bmax. Then, all the \bmax 
probability can be shifted from \jj to the a pure strategy in \calh, i.e. a pure
strategy that yields the highest payoff, and strictly increase the utility of 
the player. Thus, the strategy \jj is played with zero probability and the claim 
follows.
\qed
\end{proof}

In what follows assume that $\call \neq \emptyset$, hence $\bmax > 0$. From 
Lemma~\ref{lem:linfbr} follows that there is a best response where the strategy 
with the highest payoff is played with probability $\bp_1 +\bmax$. Hence, it can 
be shifted up to \bmax probability from pure strategies with lower payoff to 
each pure strategy with higher payoff, starting from the second pure strategy etc.
After this shift of probabilities there will be a set of pure strategies that 
where each one is played with probability $\bp_i + \bmax$ and possibly one pure
strategy $j$ that is played with probability less or equal to $\bp_j$. The question 
is whether more probability should be shifted from the low payoff strategies to
strategies that yield higher payoff. The next lemma establishes that 
no pure strategy form \call is played with positive probability in any best 
response against \profy.

\begin{lemma}
\label{lem:nofromlow}
In every best response against strategy \profy all pure strategies $i \in \call$ 
are played with zero probability.
\end{lemma}
\begin{proof}
Let $K$ denote denote the set of pure strategies that are played with positive
probability after the first shifting of probabilities. Without loss of generality
assume that each strategy $i \in K$ is played with probability $\bp_i + \bmax$.
Then the utility of the player under this strategy is equal to 
$U = \sum_{i \in K}(\bp_i+\bmax)\cdot R_i\profy - d_r\cdot \bmax$. For the sake
of contradiction, assume that there is one strategy \jj from \call that belongs 
to $K$. Suppose that probability $\delta$ is shifted from the strategy \jj to 
the first pure strategy. Then the utility for the player is equal to 
$U +\delta(R_1\profy-R_{\jj}\profy - d_r) > U$, since by definition of \call 
$R_1\profy-R_{\jj}\profy - d_r > 0$. Thus, the utility of the player is 
increasing if probability is shifted. Notice that the analysis holds even if 
the penalty is $\bmax+\delta$ instead of $\bmax$, thus the claim follows.
\qed
\end{proof}

Thus, all the probability \calp from strategies from \call should be shifted to 
strategies yield higher payoff. The question now is what is the optimal way to 
distribute that probability over the strategies with the higher payoff. Clearly, 
the same amount of probability should be shifted in all strategies in \calh since
it makes the penalty smaller. Furthermore, it is easy to see that the maximum 
amount of probability is shifted to strategies in \calh. Next we prove that if
$\calp \geq \bmax \cdot(|\calh| + |\calm|)$ then \calp is uniformly distributed over
the pure strategies in $\calh \cup \calm$.
\begin{proof}
\label{lem:linfuniform}
If $\calp \geq \bmax \cdot(|\calh| + |\calm|)$ then there is a best response 
where the probability \calp is uniformly distributed over the pure strategies in 
$\calh \cup \calm$.
\end{proof}
\begin{proof}
Let $|\calh| + |\calm| = k$ and $S = \calp - k\cdot\bmax$. Let 
$$U = \sum_{i \in \calh \cup \calm}(\bp_i+\bmax+ \frac{S}{k})R_i\profy - d_r(\bmax+ \frac{S}{k}))$$
be the utility when the probability $S$ is distributed uniformly over all pure
strategies in $\calh \cup \calm$. Furthermore, let $U'$ be the utility when 
$\delta>0$ probability is shifted from a pure strategy $j$ to the first pure 
strategy that yields the highest payoff. 
Then $U' = U + \delta(R_1\profy - R_j\profy - d_r)$, but 
$R_1\profy - R_j\profy - d_r \leq 0$ since $j \in \calh \cup \calm$. The claim 
follows.
\qed
\end{proof}

Using the previous analysis the correctness of the algorithm follows.

Note that, using similar arguments as in Lemma~\ref{lem:l1br} the next lemma can 
be proved.
\begin{lemma}
\label{lem:linfdom}
If $d_r \geq 1$, then \bp is a dominant strategy.
\end{lemma}

Furthermore, the combination of Lemma~\ref{lem:linfdom} with the fact that best 
responses can be computed in polynomial time gives the next theorem.
\begin{theorem}
\label{thm:linfeasy}
In biased games with $L_\infty$ penalty functions and 
$\max\{ d_r, d_c\} \geq 1$, an equilibrium can be computed in polynomial time.
\end{theorem}

Again we can see that there is a connection between the equilibria of the 
distance biased game and the well supported Nash equilibria (WSNE) of the 
underlying bimatrix game.
\begin{observation}
\label{obs:linf}
Let $\biased =\big( R,C, \bbb_r(\profx, \bp), \bbb_c(\profy, \bq), d_r, d_c \big)$  
be a distance biased game with $L_\infty$ penalties and let $d := \max \{d_r, d_c\}$.
Any equilirbium of \biased is a $d$-WSNE for the bimatrix game $(R, C)$.
\end{observation}

\subsubsection{Approximation algorithm}

For the quality of approximation, we can reuse the results that we proved for
the $L_1$ penalty. Lemma~\ref{lem:row_regret} applies unchanged. For
Lemma~\ref{lem:col_l1inf}, we observe that $d_c\cdot \bbb_c(\profy^*,\bq) \leq
1$ when the penalty $\bbb_c(\profy^*,\bq)$ is the $L_\infty$ norm, since for
this case it holds $\|\profy^* - \bq\|_\infty \leq 1$ and it is assumed that
$d_c \leq 1$. Thus, we have the following theorem.

\begin{theorem}
\label{thm:linfapx}
In biased games with $L_\infty$ penalties a 2/3-equilibrium can be computed in
polynomial time.
\end{theorem}

\section{Conclusions}
\label{sec:conclusions}

We have studied games with infinite action spaces, and non-linear payoff
functions. We have shown that Lipschitz continuity of the payoff function can be
exploited to provide algorithms that find approximate equilibria. For Lipschitz
games, we showed that Lipschitz continuity of the payoff function allows us to
provide an efficient algorithm for finding approximate equilibria. For penalty
games, the Lipschitz continuity of the penalty function allows us to provide a
QPTAS. Finally, we provided strongly polynomial approximation algorithms for
$L_1$, $L_2^2$, and $L_\infty$ distance biased games.

Several open questions stem from our paper.  The most important one is to
understand the exact computational complexity of equilibrium computation in
Lipschitz and penalty games. Although Theorem~\ref{thm:nofptas} states that
there no FPTAS for penalty games, the result holds only for games with penalty
functions that depend on the size of the game and tend to zero as the size
grows. Another interesting feature is that we cannot verify efficiently in all
penalty games whether a given strategy profile is an equilibrium, and so it
seems questionable whether \PPAD can capture the full complexity of penalty
games. On the other side, for the distance biased games that we studied in this
paper, we have shown that we can decide in polynomial time if a strategy profile
is an equilibrium. Is the equilibrium computation problem \PPAD-complete for the
two classes of games we studied? Are there any subclasses of penalty games, e.g.
when the underlying normal form game is zero sum, that are easy to solve?

Another obvious direction is to derive better polynomial time approximation
guarantees under for biased games. We believe that the optimization approach
used by~\citet{TS} and~\citet{DFSS14} might tackle this problem. Under the $L_1$ penalties
the analysis of the steepest descent algorithm may be similar to~\citet{DFSS14}
and therefore we may be able to obtain a constant approximation guarantee
similar to the bound of $0.5$ that was established in that paper. The other
known techniques that compute approximate Nash equilibria~\cite{BBM10} and
approximate well supported Nash equilibria~\cite{KS,FGSS12,CDFFJS15} solve a
zero sum bimatrix game in order to derive the approximate equilibrium, and there
is no obvious way to generalise this approach in penalty games. 

\bibliographystyle{ACM-Reference-Format-Journals}
\bibliography{references}

\newpage 
\appendix

\section{Proof that $\profy^{*^T}\profy^* - 2\profy_k^*\bq_k \leq 1 - 2\bq_k$.}
\label{app:proof}

%
\begin{proof}
Notice from~\eqref{eq:xi} that for all $i$ we get 
$\profy_i = \profy_k + \frac{1}{2}(\alpha_i - \alpha_k)$. Using that we can 
write the term $\profy^T\profy = \sum_i \profy_i^2$ as follows for a when \profy
has support size $s$
\begin{align*}
\sum_{i=1}^s \profy_i^2 & = \profy_i^2 + \sum_{i\neq k} \profy_i^2 \\
& = \profy_k^2 + \sum_{i\neq k} \left(\profy_k +\frac{1}{2}(\alpha_i - \alpha_k)\right)^2\\
& = s\profy^2_k + \Big(\sum_{i \neq k}(\alpha_i - \alpha_k)\Big)\profy_k 
+ \frac{1}{4} \sum_{i \neq k}(\alpha_k - \alpha_i)^2.
\end{align*}

Then we can see that $\profy^{*^T}\profy - 2\profy^{*^T}_k\bq_k$ is increasing 
as $\profy^*_k$ increases, since we know from Lemma~\ref{lem:pos-prob} that 
$\profy^*_k > 0$. This becomes clear if we take the partial derivative of 
$\profy^{*^T}\profy^* - 2\profy_k^*\bq_k$ with respect to $\profy^*_k$ which is equal to 
\begin{align*}
2s\profy^*_k + \sum_{i \neq k}(\alpha_i - \alpha_k) - 2\bq_k & = 
2s\profy^*_k + \sum_{i \neq k}2(\profy^*_i - \profy^*_k) - 2\bq_k \quad 
\text{\big(since $\profy_i = \profy_k + \frac{1}{2}(\alpha_i - \alpha_k)$\big)} \\
& = 2s\profy^*_k + 2\sum_{i \neq k}\profy^*_i - 2(s-1)\profy^*_k - 2\bq_k \\
& = 2\sum_{i=1}^s\profy^*_i - 2\bq_k \\
& = 2 - 2\bq_k \\
& \geq 0 \quad \text{(since $\profy^*_k>0$)}.
\end{align*}
Thus, the value of $\profy^{*^T}\profy^* - 2\profy_k^*\bq_k$ is maximized when 
$\profy_k^*=1$ and our claim follows.
\qed
\end{proof}

\end{document}